\documentclass[a4paper,10pt]{article}

\usepackage[utf8]{inputenc}
\usepackage{hyperref}
\usepackage{graphicx}
\usepackage[english]{babel}
\usepackage{amsmath,amssymb,amsthm,dsfont}
\usepackage{enumitem}
\usepackage{todonotes}
\usepackage{booktabs}
\usepackage{natbib}
\usepackage{authblk}
\usepackage{tikz}
\usepackage{subcaption}

\usepackage{color}

\theoremstyle{plain}
\newtheorem{theorem}{Theorem}[section]
\newtheorem{proposition}[theorem]{Proposition}
\newtheorem{corollary}[theorem]{Corollary}
\newtheorem{lemma}[theorem]{Lemma}

\theoremstyle{definition}
\newtheorem{definition}[theorem]{Definition}

\newtheorem{notation}[theorem]{Notation}
\newtheorem{idea}[theorem]{Idea}

\theoremstyle{remark}

\newtheorem{remark}[theorem]{Remark}

\makeatletter

\let\c@table\c@figure
\makeatother 

\newcommand{\R}{\mathbb{R}}
\newcommand{\E}{\mathds{E}}
\newcommand{\F}{\mathcal{F}}

\title{Unifying the theory of storage and the risk premium by an unobservable intrinsic electricity price}

\author[1,2]{W.J. Hinderks\thanks{Corresponding author: \url{wieger.hinderks@itwm.fraunhofer.de}}}
\author[1,2]{R. Korn}
\author[1,3]{A. Wagner}

\affil[1]{\small TU Kaiserslautern, Erwin-Schrödinger-Straße 1, 67663 Kaiserslautern, Germany}
\affil[2]{\small Fraunhofer ITWM, Fraunhofer-Platz 1, 67663 Kaiserslautern, Germany}
\affil[3]{\small Karlsruhe University of Applied Sciences, Faculty of Management Science and Engineering, Moltkestraße 30, 76133 Karlsruhe, Germany}

\date{\today}

\begin{document}

\maketitle

\begin{abstract}
\noindent In this paper we introduce a new concept for modelling electricity prices through the introduction of an unobservable intrinsic electricity price~$p(\tau)$. We use it to connect the classical theory of storage with the concept of a risk premium.
We derive prices for all common contracts such as the intraday spot price, the day-ahead spot price, and futures prices. Finally, we propose an explicit model from the class of structural models and conduct an empirical analysis, where we find an overall negative risk premium.

\vspace{1em}

\noindent {\bfseries Keywords:} intrinsic electricity price, theory of storage, risk premium, risk-neutral measure, real-world measure, Esscher transform
\end{abstract}


\section{Introduction} \label{section:Introduction}
Electricity is different than other underlyings of financial contracts: it is not storable. This means that electrical energy at one time point cannot be transferred to another. As a consequence power contracts with disjoint delivery time spans basically have a different underlying \citep{Hinz2005}. Of course, their prices are not necessarily uncorrelated since the price driving processes of electricity production are (auto)correlated.

Because of this non-storability of electricity the relation between spot and forward contracts is not obvious. In the literature several theories have been proposed to explain the relation between spot and forward prices for commodities. The two main theories are the \emph{theory of storage} and the concept of a \emph{risk premium}, both of which we discuss in Section~\ref{section:LiteratureReviewOnForwardPricing}. With this unclear relation between spot and forward prices also comes a lack of knowledge on what the risk-neutral measure~$Q$ should be for electricity markets. This paper uses the concept of the actual \emph{intrinsic price} of electricity, which connects the theory of storage and the concept of a risk premium.

In this paper we 
\begin{itemize}[nosep]
\item introduce a new approach to modelling electricity prices,
\item show how this approach is related to existing modelling approaches such as the Heath-Jarrow-Morton~(HJM) approach,
\item investigate the relation between the real-world measure~$P$ and the risk-neutral measure~$Q$,
\item connect our theory to the theory of storage and the concept of a risk premium,
\item and apply this theory to market data.
\end{itemize}
Section~\ref{section:LiteratureReviewOnForwardPricing} is concerned with a literature review of both main theories on forward pricing and introduces the general idea of the intrinsic price modelling approach. The mathematical theory of the intrinsic electricity price is introduced in Section~\ref{section:TheoryOfIntrinsicPrice}, whereas Section~\ref{section:ModelChoiceAndEmpiricalResults} assumes an explicit model and applies it to real data. We will see there that the risk premium is in general negative, which is in accordance with the findings of \citet{BenthCartea2008}. With this concept we connected the construction of forward curves such as given by \citet{Caldana2017} and the HJM approaches such as given by \citet{Kiesel2009,Hinz2005,Hinderks2018}.

\section{Literature review} \label{section:LiteratureReviewOnForwardPricing}
If we consider electricity delivered during a period~$\tau$, we can trade in electricity contracts for this delivery time on four markets:
\begin{itemize}[nosep]
\item the intraday spot market,
\item the day-ahead spot market,
\item the futures market,
\item and the market for options (on futures).
\end{itemize}
This market setting is summarised in Figure~\ref{fig:overviewofmeasurechanges}. The intraday market is the last market to open and is traded in (approximately) the last 24 hours before delivery. The day-ahead market is an auction, which is held one day before delivery. On the futures market, futures on the day-ahead spot price are traded up to several years before delivery. On the options market regular European call and put options on the futures contracts are available.

Figure~\ref{fig:overviewofmeasurechanges} also illustrates the probability measures usually connected to each market. Here, we denote that usually the day-ahead spot market is modelled under the \emph{real-world} measure $P$ and that derivatives' prices are computed through conditional expectation under the \emph{risk-neutral} measure~$Q$. Since intraday spot markets have only been gaining proper liquidity fairly recently, literature on stochastic modelling of intraday prices has not matured yet and the interdependence of the intraday and day-ahead spot markets is not clear.

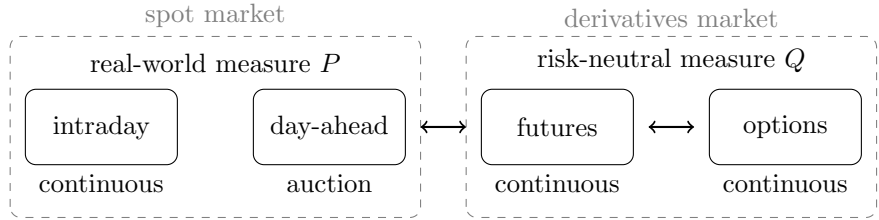
\begin{figure}[t]
\centering
\begin{tikzpicture}
\draw[rounded corners] (-7, -0.5) rectangle (-5, 0.5) {};
\node[] at (-6, 0) {intraday};
\node[below] at (-6, -0.5) {continuous};

\node[above] at (-4.5, 0.6) {real-world measure $P$};

\draw[rounded corners] (-4, -0.5) rectangle (-2, 0.5) {};
\node[] at (-3, 0) {day-ahead};

\node[below] at (-3, -0.5) {auction};

\draw[thick,<->] (-1.8,0) -- (-1.2,0);

\draw[rounded corners] (-1, -0.5) rectangle (1, 0.5) {};
\node[] at (0, 0) {futures};
\node[below] at (0, -0.5) {continuous};

\node[above] at (1.5, 0.6) {risk-neutral measure $Q$};
\draw[thick,<->] (1.8,0) -- (1.2,0);

\draw[rounded corners] (2, -0.5) rectangle (4, 0.5) {};
\node[] at (3, 0) {options};
\node[below] at (3, -0.5) {continuous};

\draw[rounded corners,dashed,gray] (-7.2, -1.2) rectangle (-1.8, 1.2) {};
\node[above,gray] at (-4.5, 1.2) {spot market};

\draw[rounded corners,dashed,gray] (-1.2, -1.2) rectangle (4.2, 1.2) {};

\node[above,gray] at (1.5, 1.2) {derivatives market};

\end{tikzpicture}
\caption{Electricity markets and the relevant probability measures. Day-ahead contracts are the underlying for the futures. The futures in turn are the underlying for the options. These derivative relations are denoted with the double-headed arrows.}
\label{fig:overviewofmeasurechanges}
\end{figure}

In the rest of this section we will write~$S(\tau)$ for the day-ahead spot price delivering 1 MW from $\tau$ to $\tau + 1$ hour and denote the price at time~$t$ of a forward on $S(\tau)$ by $f_t(\tau)$. The relation between $P$ and $Q$ -- or in other words, the relation between the spot and futures markets -- is not straightforward, since electricity is not one-dimensional in time as discussed in the \nameref{section:Introduction}. In the literature two main methods for pricing forward contracts can be found: the theory of storage and the concept of a risk premium, which we detail in the next two sections. Furthermore, we introduce a new idea using the notion of an unobservable intrinsic electricity price to model the relation between spot and forward markets.

\subsection{Theory of storage}
The theory of storage -- as its name suggests -- is based on the fact that one can buy the underlying for a forward now and sell it later \citep{Hull2000,Fama1987}. 

At time~$t$ a storable commodity can be bought at the spot market for a price~$S(t)$ and it can be held until delivery time~$\tau > t$. Comparing this strategy to that of entering a forward contract at time~$t$, which delivers the commodity at time~$\tau$, it is easy to see that the forward price should equal
\[
f_t(\tau) = e^{(r + u - y) (\tau-t)} \, S(t),
\]
where $r$ is the interest rate, $u$ corresponds to the storage costs and $y$ to the convenience yield.\footnote{The convenience yield is the implied gain of physically holding a consumption asset.} 

As said in the Introduction electricity cannot be stored and held like regular commodities such as gold. Therefore, this approach, which is based on a buy-and-hold replication strategy, cannot be used for electricity prices.

\subsection{Risk premium}
As discussed by \citet{Fama1987} there is another line in pricing commodity forwards, which introduces the concept the so-called risk premium. The risk premium at time~$t$ for delivery time~$\tau$ is defined as the difference
\begin{equation} \label{eq:DefinitionRiskPremium}
\pi_t(\tau) :=  f_t(\tau) - \E_P[S(\tau) \, | \, \F_t].
\end{equation}
The motivation behind this premium is that the difference between the futures price and the current spot price should equal the risk premium~$\pi_t(\tau)$ plus the expected difference of the future and current spot price, i.e.
\[
f_t(\tau) - S(t) = \pi_t(\tau) + \E_P[S(\tau) - S(t) \, | \, \F_t].
\]
Rewriting this yields Equation~\eqref{eq:DefinitionRiskPremium}. A common approach in electricity modelling is to assume\footnote{Or derive an equivalent measure~$Q$ from the spot price model under $P$.} that there is an equivalent measure~$Q$ such that
\[
f_t(\tau) := \E_Q[S(\tau) \, | \, \F_t],
\]
see \citet{BenthCartea2008}, for example. The risk premium then becomes
\begin{align}
\nonumber \pi_t(\tau) &=  \E_Q[S(\tau) \, | \, \F_t] -  \E_P[S(\tau) \, | \, \F_t] \\
&=  \E_P\left[ \left( \tfrac{\nu_\tau}{\nu_t} - 1 \right) S(\tau) \, | \, \F_t\right], \label{eq:RiskPremiumThroughMeasureChange}
\end{align}
where $\nu_t = \frac{dQ}{dP} \large|_{\F_t}$ is the Radon-Nikodym derivative.

\begin{remark}[Martingale property]
Usually, when we speak of \emph{the} risk-neutral measure we mean the unique equivalent measure~$Q$ such that all discounted tradable assets are martingales, i.e.
\[
e^{-r t} S(t) \overset{!}{=} \E_Q[e^{-r \tau } S(\tau) \, | \, \F_t] .
\]
However, since $S(t)$ and $S(\tau)$ basically have different underlying commodities and $S(\tau)$ is not traded at time~$t$, \citet{Benth2008} argue that this relation should not hold for \emph{a} risk-neutral measure in the electricity markets. This allows any equivalent measure to be called a pricing or risk-neutral measure.
\end{remark}

There exist several studies investigating the risk premium for electricity contracts, e.g. \citet{Redl2012,BenthCartea2008,Benth2009,Lucia2011,Viehmann2011}. However, it is hard to investigate the risk premium in the case of electricity since $S(t)$ and $S(\tau)$ basically have different underlying commodities. The method conducted by \citet[Equations (6) and (7)]{Fama1987} on a variety of different storable commodities is therefore not applicable in the electricity setting.

\citet{Redl2012,Viehmann2011} concentrate on the risk premium in the German market. They view the so-called \emph{ex post} premium, expressed as
\begin{align*}
f_t(\tau) - S(\tau) &= \left( f_t(\tau) - \E_P[S(\tau) \, | \, \F_t] \right) - \left( S(\tau) -\E_P[S(\tau) \, | \, \F_t] \right) \\
&=: \pi_t(\tau) - \varepsilon_t(\tau),
\end{align*}
where $\varepsilon_t(\tau) \in \F_\tau$ is a random variable with $P$-expectation equal to zero. Both studies find that the risk premium is positive in mean. However,their analysis is conducted by comparing futures prices with the realized spot prices and, therefore, the error terms~$\varepsilon_t(\tau)$ are assumed to be independent, which they might not be. In this case the result does not tell us anything about the risk premium, but about the average risk premium plus error term.

\citet{BenthCartea2008} define an arithmetic multi-factor model for the spot price~$S(t)$ and define a measure change from~$P$ to $Q$ with the Esscher transform to price futures contracts. They derive Equation~\eqref{eq:RiskPremiumThroughMeasureChange} in their setting and apply their model to German market as well. However, they find that the majority of the contracts has a negative risk premium. This contradicts the findings of \citet{Redl2012,Viehmann2011}.

In recent work a zero risk premium, i.e. $P = Q$, has been discussed for certain purposes such as constructing a PFC or forecasting prices \citep{Caldana2017,Steinert2018}. Other studies do not consider a pricing measure at all and thus compute all derivatives' prices through conditional expectation under the real-world measure \citep{Lyle2009}.

In light of the above discussion we find a modelling approach that just introduces the risk premium to capture the difference between spot and futures prices not completely satisfying. This method cannot answer all the questions raised by its introduction and it is extremely hard -- if not, impossible -- to verify its existence through empirical studies in the case of electricity prices, which is indicated by the contradictory evidence of the discussed studies.

\subsection{An unobservable intrinsic price}
In this section we introduce a new perspective: all power contracts deliver electrical energy during a certain delivery period. Surely, when looking at the system as a whole, this energy must have a true price, which is unobservable and intrinsic for that delivery period. What if we model this intrinsic electricity price instead of every market separately?

As a consequence we stop using the modelling approach displayed in Figure~\ref{fig:overviewofmeasurechanges}, i.e. a system where we model each market by its own price and try to connect two markets by a measure change. Instead we assume that there is an unobservable intrinsic electricity price modelled under a fixed risk-neutral~$Q$ and assume all tradable electricity contracts to be derivatives of this intrinsic electricity price. Figure~\ref{fig:overviewunobservableintrinsicprice} illustrates this approach.

In this approach we assume that all tradable contracts have dynamics under the real-world measure~$P$. Therefore it is important to define the change of measure\footnote{Note that this is the other way around compared to classical financial markets.} from~$Q$ to $P$, such that we can use the model we defined under~$Q$. In the next section we pursue this idea further and develop a general theory for the intrinsic electricity price.

\begin{figure}[t]
\centering
\begin{tikzpicture}
\draw[rounded corners] (-7, -0.5) rectangle (-5, 0.5) {};
\node[] at (-6, 0) {intraday};
\draw[thick,<->] (-6,-0.7) -- (-6,-1.3);

\draw[rounded corners] (-4, -0.5) rectangle (-2, 0.5) {};
\node[] at (-3, 0) {day-ahead};
\draw[thick,<->] (-3,-0.7) -- (-3,-1.3);

\draw[rounded corners] (-1, -0.5) rectangle (1, 0.5) {};
\node[] at (0, 0) {futures};
\draw[thick,<->] (0,-0.7) -- (0,-1.3);

\draw[rounded corners] (2, -0.5) rectangle (4, 0.5) {};
\node[] at (3, 0) {options};
\draw[thick,<->] (3,-0.7) -- (3,-1.3);

\draw[rounded corners] (-7, -1.5) rectangle (4, -2.5) {};
\node[] at (-1.5, -2) {intrinsic electricity price};

\node[] at (4.5, 0) {$P$};
\node[] at (4.5, -2) {$Q$};

\node[above,gray] at (4.5,-1) {observable};
\draw[thick,gray,dashed] (-7,-1) -- (5,-1);
\node[below,gray] at (4.5, -1) {hidden};

\end{tikzpicture}
\caption{Change of the modelling approach of Figure~\ref{fig:overviewofmeasurechanges} to an approach with an unobservable intrinsic electricity price, which lives under the risk-neutral measure~$Q$. All products traded at the market have dynamics under the real-world measure~$P$.}
\label{fig:overviewunobservableintrinsicprice}
\end{figure}
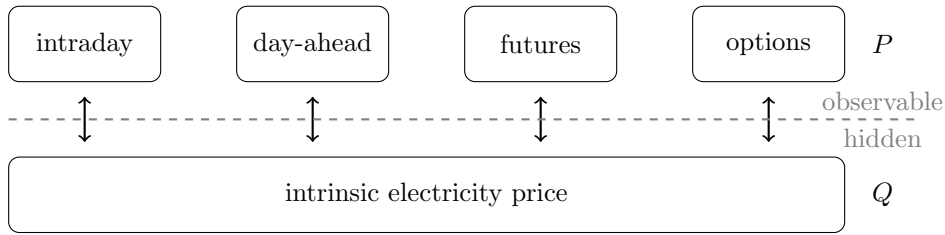

\section{The intrinsic electricity price under $Q$} \label{section:TheoryOfIntrinsicPrice}
Let $(\Omega, \mathcal{A}, Q)$ be a complete probability space. On this probability space we assume $W = \{ W_t ; t \geq 0\}$ to be a $d$-dimensional Brownian motion with augmented natural filtration~$\F = \{ \F_t ; t \geq 0 \}$. For technical convenience we assume that $\F_\infty = \mathcal{A}$. We interpret the Brownian motion~$W$ as realization of the flow of information in the electricity market. In the following we will always denote delivery time by~$\tau$ and trading time by~$t$.

\begin{notation}[Intrinsic electricity price] \label{notation:IntrinsicElectricityPrice}
We introduce the stochastic process denoted by~$p = \{ p(\tau) ; \tau \geq 0 \}$ and call it the \emph{intrinsic electricity price}. 
\end{notation}

By Notation~\ref{notation:IntrinsicElectricityPrice} we mean that $p(\tau)$ equals the average of all actual costs made by all market players to generate 1MW electricity during the delivery time interval~$[\tau$, $\tau + \varepsilon)$ with constant $\varepsilon > 0$. Basically we just introduced the notation $p(\tau) := p\left([\tau, \tau+\varepsilon)\right)$. Because $p(\tau)$ gets it value from the actual occurred costs of all electricity generated in the grid, we call it the \emph{intrinsic price}. The constant~$\varepsilon$ is meant as the delivery duration of our standard or smallest contract, which typically is an hour.\footnote{However, it can also be considered to be any other unit of time such as a quarter hour or a day.} It is clear that the actual costs are not known at the beginning of the interval $[\tau, \tau + \varepsilon)$, i.e. $p(\tau)$ is not $\F_\tau$-measurable. However, when the delivery period is over, i.e. at time $\tau+\varepsilon$, these costs are known or can be derived and, therefore, $p(\tau)$ is $\F_{\tau+\varepsilon}$-measurable.

For notational convenience we will write \emph{delivery time~$\tau$} for the delivery period~$[\tau, \tau+\varepsilon)$ throughout the rest of this paper. Furthermore, we assume our probability space to be equipped with a measure~$Q$ and call this \emph{risk-neutral measure}. The following definition validates this naming:
\begin{definition}[Tradable electricity price]
At trading time~$t$ the price of electricity for delivery time~$\tau$ is defined by
\[
p_t(\tau) := \E_Q [e^{-r(\tau + \varepsilon - t)} p(\tau) \, | \, \F_t],
\]
where $r$ is the risk-free rate. We call $p_t(\tau)$ the \emph{tradable electricity price} at (trading time)~$t$ and for delivery time~$\tau$.
\end{definition}
The tradable electricity price is unobservable and, therefore, not really tradable. However, if in a hypothetical world, electricity with delivery time~$\tau$ would be a storable commodity in the sense that one could buy electricity with delivery time~$\tau$ before the time point~$\tau$ and keep it to consume it during the delivery period~$[\tau, \tau+\varepsilon)$, the tradable electricity price would be its spot price. We do not name it the spot price, since this would cause confusion with the day-ahead and intraday spot markets. The tradable electricity price is an artificial price, to which we can apply the theory of storage. With this construction we artificially fit electricity in the framework of storable commodities.

From the definition of the filtration~$\F$ and the tradable electricity price~$p_t(\tau)$ it is clear that $p_0(\tau) = e^{-r (\tau+\varepsilon)} \E_Q p(\tau)$ and $p_{\tau + \varepsilon}(\tau) = p(\tau)$. Under the risk-neutral measure the discounted tradable assets are $Q$-martingales, i.e. for $t > s$ we have
\[
\E_Q[ e^{-rt} p_t(\tau) \, | \, \F_s] = \E_Q[\E_Q [e^{-r(\tau+\varepsilon)} p(\tau) \, | \, \F_t] \, | \, \F_s]  = e^{-rs} p_s(\tau).
\]
This is the main reason why we define the model under the risk-neutral probability measure~$Q$.

\begin{definition}[Intraday price] \label{definition:IntradayPrice}
The \emph{intraday price} for delivery time~$\tau$ is defined as $I(\tau) := p_\tau(\tau)$.
\end{definition}

In the asymptomatic case where we let the delivery length~$\varepsilon \to 0$, the intraday price tends to the real electricity price~$I(\tau) \to p(\tau)$. Throughout the rest of this paper we denote the length of one day by~$\delta$ for notational convenience.

\begin{definition}[Day-ahead spot price] \label{definition:DayAheadSpotPrice}
The \emph{day-ahead spot price} for delivery time~$\tau$ is defined as $S(\tau) := p_{\tau - \delta}(\tau)$.
\end{definition}


Note that although we write $S(\tau)$, it is $\F_{\tau-\delta}$-measurable. Furthermore, since the stochastic process~$\{ e^{-rt} p_t(\tau) ; t \geq 0 \}$ is a $Q$-martingale by construction we find that $\E_Q[I(\tau) \, | \, \F_{\tau - \delta}] = e^{r \delta} S(\tau)$. This merely states that under the risk-neutral measure~$Q$ the expectation of the intraday price one day in advance, i.e. one day ahead, is the day-ahead spot price. Moreover, we can apply the martingale representation theorem to find:

\begin{corollary} \label{corollary:MartingaleRepresentationTheorem}
For each delivery time~$\tau$ there exists an a.s. unique, predictable, $\R^d$-valued process $\varphi(\tau) = \{\varphi_t(\tau) ; t \geq 0\}$ such that
\[
p_t(\tau) = e^{r t} p_0(\tau) +  e^{-r(\tau + \varepsilon - t)} \int_0^t \varphi_s(\tau)' \cdot dW_s
\]
for all $t \geq 0$.
\end{corollary}
\begin{proof}
This is the exact statement of the martingale representation theorem applied to our setting \citep{Protter2005}.
\end{proof}

\begin{definition}[Price generating process]
We call the a.s. unique process~$\varphi(\tau)$ from Corollary~\ref{corollary:MartingaleRepresentationTheorem} the \emph{price generating process}.
\end{definition}

From Corollary~\ref{corollary:MartingaleRepresentationTheorem} we can derive that the dynamics of the tradable electricity price are given by
\begin{equation} \label{eq:IntrinsicElectricityPriceForwardSDE}
dp_t(\tau) = r p_t(\tau) \, dt + e^{-r(\tau +\varepsilon - t)} \varphi_t(\tau)' \cdot dW_t.
\end{equation}
Furthermore, we immediately see that we have a recursive relation between the tradable electricity prices of a fixed delivery time~$\tau$: for $t \geq u \geq 0$ we have
\[
p_t(\tau) = e^{r (t - u)} p_u(\tau) +  e^{-r(\tau + \varepsilon - t)} \int_u^t \varphi_s(\tau)' \cdot dW_s.
\]
From this relation it immediately follows that:

\begin{corollary} \label{corollary:AlternativeRepresentationIntrinsicElectricityPrice}
An alternative representation of the intrinsic electricity price is
\[
p(\tau) = e^{r (\tau + \varepsilon - t)} p_t(\tau) +  \int_t^{\tau + \varepsilon} \varphi_s(\tau)' \cdot dW_s
\]
for all $\tau + \varepsilon \geq t \geq 0$.
\end{corollary}
\begin{proof}
Follows by the $\F_{\tau + \varepsilon}$-measurability of the intrinsic electricity price.
\end{proof}


As in the theory of storage we can now introduce the forward price of an electricity contract with delivery~$\tau$. We assume the storage costs~$u$ and convenience yield~$y$ to equal zero, since the electricity is not actually storable. Because the forward can only be settled at the end of the delivery period, the payment date is at $\tau + \varepsilon$ and we have to discount from that time point.

\begin{definition}[Forward price]
The \emph{forward price} is given by 
\[
f_t(\tau) := e^{r(\tau + \varepsilon -t)} p_t(\tau)
\]
for $t \geq 0$.
\end{definition}

It is clear that we have $f_t(\tau) = \E_Q [p(\tau) \, | \, \F_t]$ and thus that for fixed delivery times~$\tau$ the process $\{  f_t(\tau); t\geq 0 \}$ is a $Q$-martingale. Furthermore, from Corollary~\ref{corollary:MartingaleRepresentationTheorem} it follows that
\[
f_t(\tau) =  f_0(\tau) + \int_0^t \varphi_s(\tau)' \cdot dW_s
\]
for all $t \geq 0$.

\begin{idea}
In light of Corollary~\ref{corollary:MartingaleRepresentationTheorem} there are now two equivalent possibilities to assume an explicit model:
\begin{itemize}[nosep]
\item through the intrinsic electricity price~$p(\tau)$ and the computation of its conditional expectation,
\item or through the initial forward price~$f_0(\tau)$ (e.g. the price forward curve, PFC) and the price generating process~$\varphi(\tau)$.
\end{itemize}
We will come back to this in Section~\ref{section:ModelChoiceAndEmpiricalResults}, where we will do an empirical study.
\end{idea}

\begin{remark}[Heath-Jarrow-Morton framework]
Our approach is based on the intrinstic price~$p(\tau)$, which can only be observed after the delivery period is over. However, as a consequence of Corollary~\ref{corollary:MartingaleRepresentationTheorem} we derived the modelling approach of electricity prices through the price generating process~$\varphi$ and the initial forward curve, which usually is called a Heath-Jarrow-Morton~(HJM) approach after the famous framework introduced for interest rates by \citet{Heath1992}. In the context of electricity prices the HJM approach has been studied extensively, e.g. \citet{Hinz2005,Kiesel2009,Latini2018,Hinderks2018,Benth2019}.
\end{remark}

\subsection{Futures} \label{setion:IntrinsicPriceFuturesAndOptions}
Consider a futures contract with increasing delivery times~$\mathcal{T} := \{\tau_1, \tau_2, \dots, \tau_n\}$, i.e. $0 \leq \tau_1 < \tau_2 < \dots < \tau_n$, and financial fulfillment at final delivery~$\tau_n$. Since in the electricity market futures are settled against the spot price, the pay-off at~$\tau_n$ is given by $\sum_{i = 1}^n S(\tau_i)$. It follows that the price of a futures contract is given by
\[
F_t(\mathcal{T}) := \frac{1}{n} \E_Q\left[\sum_{i = 1}^n S(\tau_i) \, \Big| \, \F_t\right] = \frac{e^{-r(\delta + \varepsilon)}}{n} \sum_{i = 1}^n  f_{t \wedge (\tau_i - \delta)}(\tau_i)
\]
for all $t \geq 0$.

\begin{theorem}
The futures price process~$\{  F_t(\mathcal{T}); t\geq 0 \}$ is a $Q$-martingale.
\end{theorem}
\begin{proof}
The statement holds since the futures price is the weighted sum of $n$ stopped $Q$-martingales.
\end{proof}

From the definition of the tradable electricity price it is immediately clear that for all times $ 0 \leq t \leq \tau_1 - \delta$ the price of a futures is given by
\[
F_t(\mathcal{T}) =  \E_Q\left[ p(\mathcal{T}) \, | \, \F_t \right],
\]
where $ p(\mathcal{T}) :=  \tfrac{1}{n}  e^{-r(\delta + \varepsilon)} \sum_{i = 1}^n p(\tau_i)$. Furthermore, with the help of Corollary~\ref{corollary:MartingaleRepresentationTheorem} we can equivalently write for all times $ 0 \leq t \leq \tau_1 - \delta$
\[
F_t(\mathcal{T}) = F_0(\mathcal{T})  +  \int_0^t \varphi_s(\mathcal{T})' \cdot dW_s,
\]
where we define $\varphi_s(\mathcal{T}) := \tfrac{1}{n}  e^{-r(\delta + \varepsilon)} \sum_{i = 1}^n \varphi_s(\tau_i)$.

\subsection{Real-world measure $P$} \label{section:IntrinsicPriceMeasureChange}
Since the prices of the traded products move under the real-world measure~$P$, cf. Figure~\ref{fig:overviewunobservableintrinsicprice}, we need a to change to this measure to simulate the intrinsic process. In this section we assume that we change from the risk-neutral measure~$Q$ to the real-world measure~$P$ by its Radon-Nikodym derivative, i.e. 
\[
\nu_t := \frac{dP}{dQ} \Big|_{\F_t}
\]
for all $t \geq 0$. It is common to use the stochastic exponential to define the Radon-Nikodym derivative:
\begin{definition}
For an adapted $\R^d$-valued process~$\theta = \{ \theta_t ; t \geq 0 \}$ we define the Radon-Nikodym by
\[
\nu_t := \exp \left( \int_0^t \theta_s' \cdot dW_s - \frac{1}{2} \int_0^t \theta_s' \cdot \theta_s \, ds \right),
\]
i.e. by the stochastic exponential of $\int_0^t \theta_s' \cdot dW_s$.
\end{definition}

We assume that the Novikov condition is fulfilled, i.e.
\[
\E_Q\left[ e^{-\frac{1}{2} \int_0^t \theta_s' \cdot \theta_s \, ds} \right] < \infty
\]
for all $t \geq 0$. The Girsanov theorem then tells us that $\tilde{W}_t := W_t - \int_0^t \theta_s \, ds$ is a Brownian motion under $P$, cf. \citet{Korn2001}. Using this Brownian motion we can rewrite the tradable electricity price as
\[
p_t(\tau) = e^{r t} p_0(\tau) + e^{-r(\tau +\varepsilon - t)} \int_0^t \varphi_s(\tau)' \cdot \theta_s \, ds +  e^{-r(\tau +\varepsilon- t)} \int_0^t \varphi_s(\tau)' \cdot d\tilde{W}_s
\]
under $P$.

Since we consider the real-world measure~$P$ and the risk-neutral measure~$Q$ to be two different measures, it follows that we can also define a risk premium in this setting as defined in Equation~\eqref{eq:DefinitionRiskPremium}:
\begin{definition}[Risk premium]
We call the $\F_t$-measurable random variable
\[
\pi_t(\tau) := f_t(\tau) - \E_P\left[ p_\tau(\tau) \, | \, \F_t \right]
\]
the \emph{risk premium} for delivery time~$\tau$.
\end{definition}

Recall that $p_t(\tau)$ is the unobservable tradable electricity price and plays the same role in our theory as the spot price of storable commodities. The risk premium can alternatively be written as
\begin{align*}
\pi_t(\tau) &=  \E_Q[p(\tau) \, | \, \F_t] -  \E_P[  \E_Q[p(\tau) \, | \, \F_\tau] \, | \, \F_t] \\
&=  \E_Q\left[ \left(1 - \frac{\nu_\tau}{\nu_t} \right)p(\tau) \, | \, \F_t\right],
\end{align*}
Note that here we change from $P$ to $Q$ instead of the other way around, which is more common in financial mathematics. 

\begin{theorem}
The risk premium is given by
\[
\pi_t(\tau) =  \E_Q\left[ \left(1 - e^{\int_t^\tau \theta_s' \cdot dW_s - \frac{1}{2} \int_t^\tau \theta_s' \cdot \theta_s \, ds } \right)   \int_t^{\tau + \varepsilon} \varphi_s(\tau)' \cdot dW_s \right]
\]
for all $t \leq \tau + \varepsilon$.
\end{theorem}
\begin{proof}
We use Corollary~\ref{corollary:AlternativeRepresentationIntrinsicElectricityPrice} to see that for $t \leq \tau + \varepsilon$ we have
\[
p(\tau) =  e^{r t} p_t(\tau) + \int_t^{\tau+\varepsilon} \varphi_s(\tau)' \cdot dW_s,
\]
where the first term is $\F_t$-measurable and the second term is independent of $\F_t$. Now we directly compute
\[
\pi_t(\tau) =  \E_Q\left[ \left(1 - e^{\int_t^\tau \theta_s' \cdot dW_s - \frac{1}{2} \int_t^\tau \theta_s' \cdot \theta_s \, ds } \right) p(\tau) \, | \, \F_t\right]
\]
where the result follows by plugging in the representation of $p(\tau)$ that we just derived.
\end{proof}

The interpretation of the above theorem is clear: the risk premium is the expected uncertainty left in the intrinsic price, i.e. the integral over the price generating process from $t$ to $\tau$, weighted with the change induced through the measure change.

\section{Explicit model choice and empirical results} \label{section:ModelChoiceAndEmpiricalResults}
In this section we assume an explicit model for the intrinsic electricity price~$p(\tau)$ by using a structural model approach. Section~\ref{section:ExplicitStructuralModelChoice} proposes the explicit model and Section~\ref{section:ExplicitModelResults} discusses its empirical results. The goal of this section is merely to give an example of what can be done within the framework of the intrinsic electricity price.

\subsection{Structural model} \label{section:ExplicitStructuralModelChoice}
Structural models have their roots in the work of \citet{Barlow2002} and there have been many studies extending this idea, e.g. \citet{Aid2009,Lyle2009,Wagner2014}. As in \citet{Wagner2014} we assume that the ex post\footnote{With ex post we mean that the system load $G_{\tau+\varepsilon}$ is the system load for the delivery period from $\tau$ to $\tau+ \varepsilon$.} \emph{system load} or \emph{system generation}\footnote{The system demand and system generation are always balanced, therefore we can take either one.}~$G_\tau$ is defined by
\[
G_\tau := g(\tau) + X_\tau,
\]
where $g(\tau)$ is a deterministic seasonality function capturing all cyclic and seasonal behaviour and $X_\tau$ is a Gaussian Ornstein-Uhlenbeck~(OU) process. The mean-reverting process~$X_\tau$ is the solution of the following stochastic differential equation under $Q$:
\[
dX_\tau =   -\lambda  X_\tau \, d\tau + \sigma \, dW_\tau, \quad X_0 = x_0 \in \R
\]
where $W$ is a one-dimensional Brownian motion and $\lambda > 0$, $\sigma > 0$, and $\mu$ are real-valued model parameters. Its strong solution is given by
\[
X_\tau = e^{-\lambda \tau} x_0  + \int_0^\tau \sigma e^{-\lambda(\tau  - s)} \, dW_s.
\]
Recall that $\varepsilon > 0$ is the duration of the delivery period, which is fixed. As an auxiliary time variable we define ex post delivery time~$\tau_e := \tau + \varepsilon$. Using the system load as in the structural modes of \citet{Wagner2014} we can define the intrinsic electricity price as
\begin{equation} \label{eq:EmpiricalStructuralModelIntrinsicElectricityPrice}
p(\tau) := e^{\alpha_1 (G_{\tau_e} - \beta_1 )} - e^{\alpha_2 (G_{\tau_e} - \beta_2 )} + \gamma_3(\tau),
\end{equation}
where $\alpha_1 >0$, $\alpha_2 < 0$, $\beta_1$, and $\beta_2$ are real-valued parameters, and $\gamma_3(\tau)$ is a deterministic function. With the help of the auxiliary process
\[
\gamma_i(t; \tau) := \exp \left\{ \alpha_i \left(  g(\tau_e)  +  e^{-\lambda(\tau_e - t)} X_t + \frac{\alpha_i \sigma^2}{4 \lambda} \left(1 - e^{-2\lambda(\tau_e - t)}\right) - \beta_i \right) \right\}
\]
for $i =1, 2$, we can derive the tradable electricity price:
\begin{lemma}[Tradable electricity price] \label{lemma:EmpiricalIntrinsicElectricityPriceTradablePrice}
The tradable electricity price is given by
\[
p_t(\tau) = e^{-r(\tau_e -t)} \left( \gamma_1(t; \tau) - \gamma_2(t; \tau) + \gamma_3(\tau) \right)
\]
for all $ t \leq \tau_e$.
\end{lemma}
\begin{proof}
Using the fact that
\[
X_{\tau_e} = e^{-\lambda (\tau_e - t)} X_t  +  \int_t^{\tau_e} \sigma e^{-\lambda(\tau_e - s)} \, dW_s,
\]
we see that 
\[
\E_Q\left[ e^{\alpha_i X_{\tau_e}} \, | \, \F_t \right] = e^{\alpha_i e^{-\lambda (\tau_e - t)} X_t  } \E_Q \left[ e^{\alpha_i \int_t^{\tau_e} \sigma e^{-\lambda(\tau_e - s)} \, dW_s}\right]
\]
for $i = 1, 2 $. From this the result follows by explicit computation of the expectation of the lognormal distribution.
\end{proof}

It follows directly that
\[
p_t(\tau) = e^{rt} p_0(\tau) + e^{-r(\tau_e-t)} \left\{ \left[\gamma_1(t; \tau) - \gamma_1(0; \tau)\right] - \left[\gamma_2(t; \tau) - \gamma_2(0; \tau) \right]\right\}.
\]
and in particular
\[
f_t(\tau) = \gamma_1(t; \tau) - \gamma_2(t; \tau) + \gamma_3(\tau)
\]
for all $t \leq \tau_e$. From the above equation we can derive the price generating process with the help of Theorem~\ref{corollary:MartingaleRepresentationTheorem}:
\begin{proposition}[Price generating process]
The price generating process process is given by
\[
\varphi_t(\tau) = 
\begin{cases}
\sigma e^{-\lambda (\tau_e- t)} \left[ \alpha_1  \gamma_1(t; \tau) - \alpha_2 \gamma_2(t; \tau) \right], \quad & \text{if } t \leq \tau_e , \\
0, & \text{else},
\end{cases}
\]
for all $\tau  \geq 0$.
\end{proposition}
\begin{proof}
From Corollary~\ref{corollary:MartingaleRepresentationTheorem} we know that we should find $\varphi_t(\tau)$ such that
\[
\int_0^t \varphi_s(\tau) \, dW_s =  \left[\gamma_1(t; \tau) - \gamma_1(0; \tau)\right] - \left[\gamma_2(t; \tau) - \gamma_2(0; \tau) \right].
\]
We introduce an auxiliary processs
\[
dM_t = \sigma e^{-\lambda (\tau_e - t)} \, dW_t, \quad M_0 = 0,
\]
and rewrite
\[
\gamma_i(t; \tau) = e^{ \alpha_i \left( g(\tau_e) +  e^{-\lambda \tau_e} x_0 +  M_t + \frac{\alpha_i \sigma^2}{4 \lambda} \left(1 - e^{-2\lambda(\tau_e - t)}\right) - \beta_i \right) }
\]
We apply It\^o's lemma on $\gamma_i$ and $M_t$ to find that
\[
d\gamma_i  = \left( \frac{\partial}{\partial t}\gamma_i + \frac{\sigma^2}{2} e^{-2\lambda (\tau_e - t)} \, \frac{\partial^2}{\partial x^2} \gamma_i \right) dt + \sigma e^{-\lambda (\tau_e - t)}  \frac{\partial}{\partial x}\gamma_i \, dW_t.
\]
Recalling that $\alpha_i^{-2} \frac{\partial^2}{\partial x^2} \gamma_i = \alpha_i^{-1}  \frac{\partial}{\partial x} \gamma_i = \gamma_i$ and computing the derivative with respect to time
\[
\frac{\partial}{\partial t}\gamma_i =  -\frac{\alpha_i^2 \sigma^2}{2}  e^{-2\lambda(\tau_e - t)} \, \gamma_i ,
\]
then yields
\[
d\gamma_i  =  \alpha_i \sigma e^{-\lambda (\tau_e - t)} \gamma_i \, dW_t,
\]
which shows the result.
\end{proof}

\subsection{Empirical results} \label{section:ExplicitModelResults}
In this section we calibrate the model to real data. In light of our data study we want to emphasize that in this paper we our main goal was to set up the concept of the intrinsic electricity price and how it relates theoretically to the existing work. With the data study in this section we merely want to show one explicit model choice and its practical applications and effects. Therefore, this study is indifferent to the fact that the most recent market data is not available. As such the risk premium that we find in this section, is also not meant as a value for the current risk premium.

Throughout the rest of this section we assume that we measure time in hours. Therefore, we assume~$\varepsilon = 1$ and $\delta = 24$. We will evaluate contracts with delivery times of the form~$\tau = k \varepsilon$ for $k\in\mathbb{N}$. For the annual risk-free interest rate we choose~$r= 0.001$.

\begin{remark}[Data set] \label{remark:DataSet}
We have the following data from the German/Austrian market:
\begin{itemize}[nosep]
\item the hourly system load~$G_{\tau_e}$ from 1 January 2014 to 15 April 2018,
\item the hourly day-ahead spot prices~$S^M(\tau)$ and the hourly $\text{ID}_3$ prices\footnote{Since there is no unique intraday price, we assume the German intraday index $\text{ID}_3$ to be `the' intraday price.}~$I^M(\tau)$ from 28 June 2015 to 15 April 2018.
\end{itemize}
We use the whole data set for the estimation.
\end{remark}

\begin{remark}[Dynamics under $P$] \label{remark:EmpiricalSystemLoadUnderP}
Assuming that the Girsanov parameter as introduced in Section~\ref{section:IntrinsicPriceMeasureChange} is constant $\theta_t \equiv \lambda \theta \in \R$, we find that the Ornstein-Uhlenbeck process~$X_{\tau_e}$ can be rewritten under $P$ as
\[
X_{\tau} = e^{-\lambda\tau} x_0  +  \left(1 - e^{-\lambda\tau}\right) \sigma \theta + \int_0^{\tau} \sigma e^{-\lambda(\tau - s)} \, d\tilde{W}_s,
\]
where $\tilde{W}_t$ is a $P$-Brownian motion. It follows that we can split $G_\tau = \tilde{g}(\tau) + \tilde{X}_{\tau}$ under $P$, if $\tilde{X}$ is a $P$-Gaussian Ornstein-Uhlenbeck process defined by
\[
d\tilde{X}_\tau = -\lambda \tilde{X}_\tau \, d\tau + \sigma \, d\tilde{W}_\tau, \quad \tilde{X}_0 = x_0
\]
and when we define
\[
\tilde{g}(\tau) := g(\tau) + \left(1 - e^{-\lambda \tau}\right)  \sigma \theta.
\]
Assuming the mean reversion speed~$\lambda$ is small we can use the first order approximation~$1 - e^{-\lambda \tau} \approx \lambda \tau $ to find
\[
\tilde{g}(\tau) \approx g(\tau) + \lambda \sigma \theta \tau,
\]
which we will use to deseasonalize the system load~$G_\tau$ under $P$. Furthermore, in the approximated setting we have the following relation $X_\tau = \tilde{X}_\tau + \lambda \sigma \theta \tau$ between the two Ornstein-Uhlenbeck processes.
\end{remark}

\begin{figure}[p]
\centering
\begin{subfigure}[c]{\textwidth}
\includegraphics[width=\textwidth]{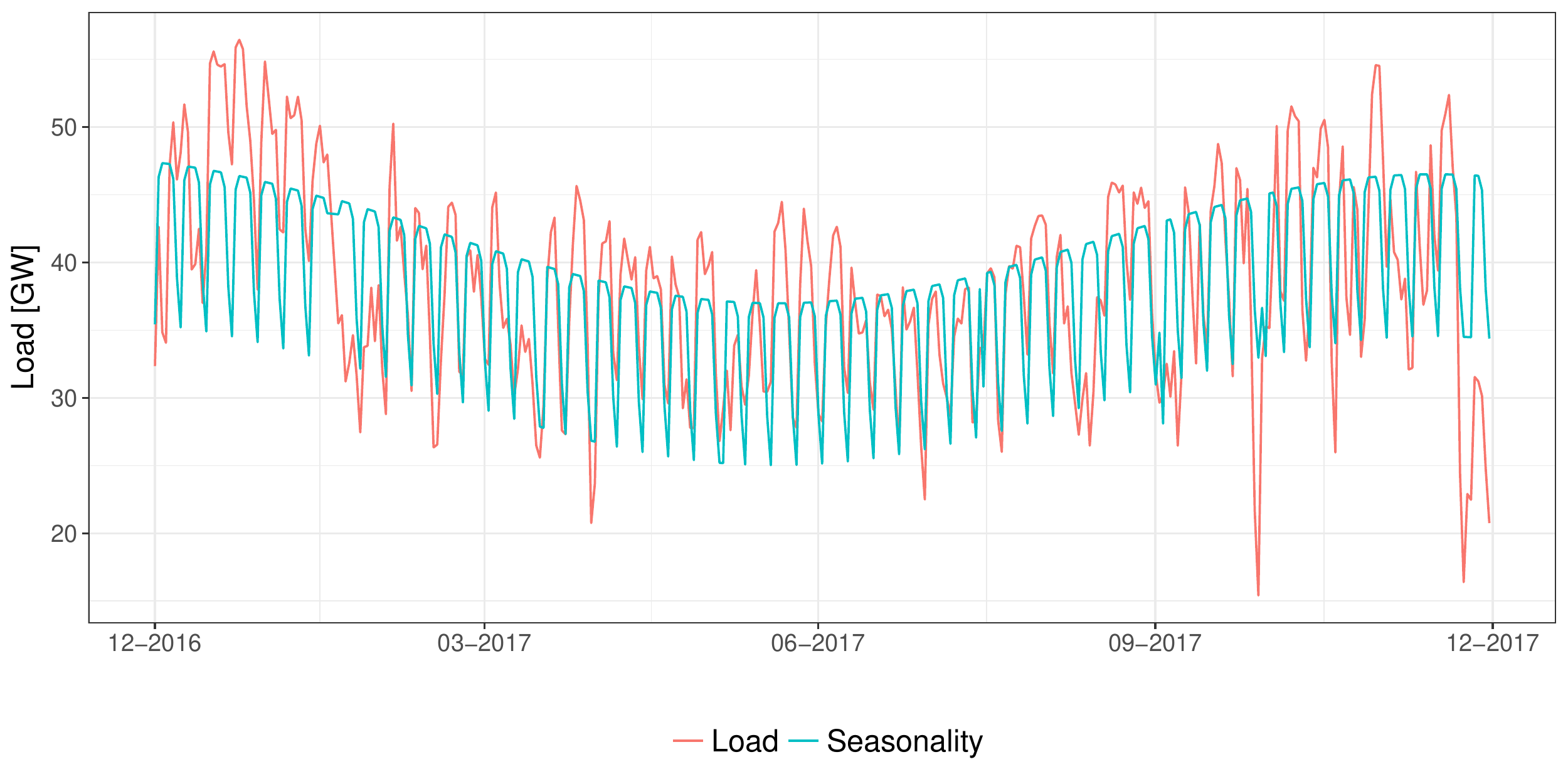}
\subcaption{Actual baseload system load~$G_\tau$ for the year 2017 together with the estimated seasonality function~$\tilde{g}(\tau)$.}
\label{fig:TotalGenerationWithSeasonality2017}
\end{subfigure}

\begin{subfigure}[c]{\textwidth}
\centering
\includegraphics[width=\textwidth]{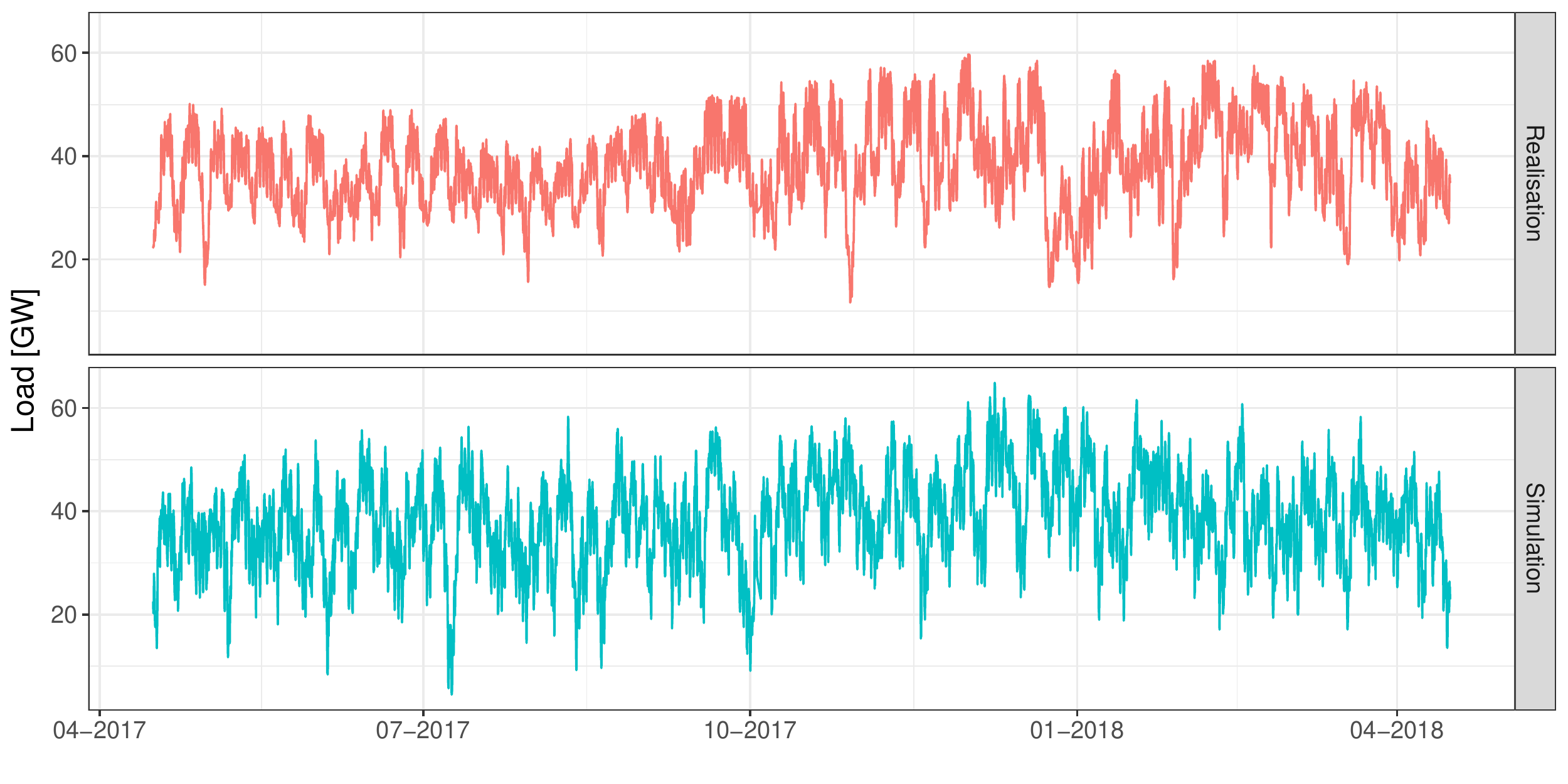}
\subcaption{Simulation of a path of the system generation~$G_\tau$ with hourly granularity for the last year of the data set, i.e. from 16 April 2017 to 15 April 2018.}
\label{fig:EmpiricalIntrinsicSystemGenerationSimulation}
\end{subfigure}
\caption{System load.}
\end{figure}

As discussed in Remark~\ref{remark:EmpiricalSystemLoadUnderP} the system load~$G_\tau$ moves under~$P$. We define the $P$-load seasonality function
\begin{equation} \label{eq:EmpiricalIntrinsicGTilde}
\tilde{g}(\tau) := z_0 +  z_1 \tau  + z_2 \sin\left( \tfrac{2 \pi}{365 \cdot 24} \tau\right) + z_3 \cos\left( \tfrac{2 \pi}{365 \cdot 24} \tau\right) + \text{DoW}_\tau + \text{HoD}_\tau,
\end{equation}
where $ \text{DoW}_\tau$ and $ \text{HoD}_\tau$ are dummy variables\footnote{This means they take a different constant value for a different \emph{day of the week}~(DoW) and \emph{hour of the day}~(HoD). Mathematically, they are just the sum of weighted indicator functions.} for the day of the week\footnote{We define four classes of weekdays: Mondays and Fridays; Tuesdays, Wednesdays, and Thursdays; Saturdays, bridge days (i.e. a days between a holiday and a weekend), and partial holidays (i.e. holidays in some but not all German federal states); Sundays and holidays.} and hour of the day. We directly estimate $\tilde{g}$ by linear least squares from the load data. Figure~\ref{fig:TotalGenerationWithSeasonality2017} shows the estimated seasonality together with the system load for the year 2017. The estimate $\tilde{g}$ can be used to deseasonalize the data $\tilde{X}_\tau = G_\tau - \tilde{g}(\tau)$, after which $\lambda$ and $\sigma$ can be estimated by maximum likelihood. The estimates of $\lambda$ and $\sigma$ are shown in Table~\ref{table:EmpiricalIntrinsicPriceParameters}. Figure~\ref{fig:EmpiricalIntrinsicSystemGenerationSimulation} illustrates a sample path of the system load~$G_\tau$ modelled with the estimated parameters.

In order to proceed with the estimation from market prices we need an estimate of the seasonality function~$\gamma_3$. We estimate the same type of formula as for~$\tilde{g}$, cf. Equation~\eqref{eq:EmpiricalIntrinsicGTilde}. We estimated $\gamma_3$ with linear least squares to a mixture of the day-ahead and intraday spot prices $\frac{I^M + S^M}{1+e^{-r\delta} }$. This corresponds approximately to the seasonality of the intrinsic price.

\begin{table}[t]
\centering
\bgroup
\def\arraystretch{1.3}
 \begin{tabular}{l | r }
 \toprule
 Parameter &  Value \\
 \midrule
$\lambda$ & 0.0298 \\
$\sigma$ & 1.4988 \\
$x_0$ & -12.5776 \\
\midrule 
$\alpha_1$ & 0.1949 \\
$\alpha_2$ & -0.1796 \\
$\beta_1$ & 43.8799 \\
$\beta_2$ & 37.4548 \\
$\theta$ & -0.0036 \\
 \bottomrule
 \end{tabular}
\egroup
 \caption{Estimated parameters of the structural model.}
 \label{table:EmpiricalIntrinsicPriceParameters}
\end{table}

We can combine the above Remark~\ref{remark:EmpiricalSystemLoadUnderP} to calibrate the supply function parameters $\alpha_1$, $\alpha_2$, $\beta_1$, and $\beta_2$ together with~$\theta$. We use the \texttt{R} function \texttt{optim} with method \texttt{BFGS} to minimize the mean squared error of the realized and theoretical day-ahead and intraday prices. The theoretical prices are given by Lemma~\ref{lemma:EmpiricalIntrinsicElectricityPriceTradablePrice}. This means that we minimize
\begin{equation} \label{eq:OptimizationProblemForCalibration}
\min_{\alpha_1, \alpha_2, \beta_1, \beta_2, \theta} \frac{1}{2N} \sqrt{\sum_{k=1}^N (I^M(k\varepsilon) - I(k\varepsilon))^2 + \sum_{k=1}^N (S^M(k\varepsilon) - S(k\varepsilon))^2},
\end{equation}
where the superscript $M$ stands for the market price. As initial parameters we used the ones obtained from fitting the intraday prices directly to the formula for the intrinsic electricity price of Equation~\eqref{eq:EmpiricalStructuralModelIntrinsicElectricityPrice}. The results of the estimation procedure are given in Table~\ref{table:EmpiricalIntrinsicPriceParameters}. 

\begin{figure}[t]
\centering
\includegraphics[width=\textwidth]{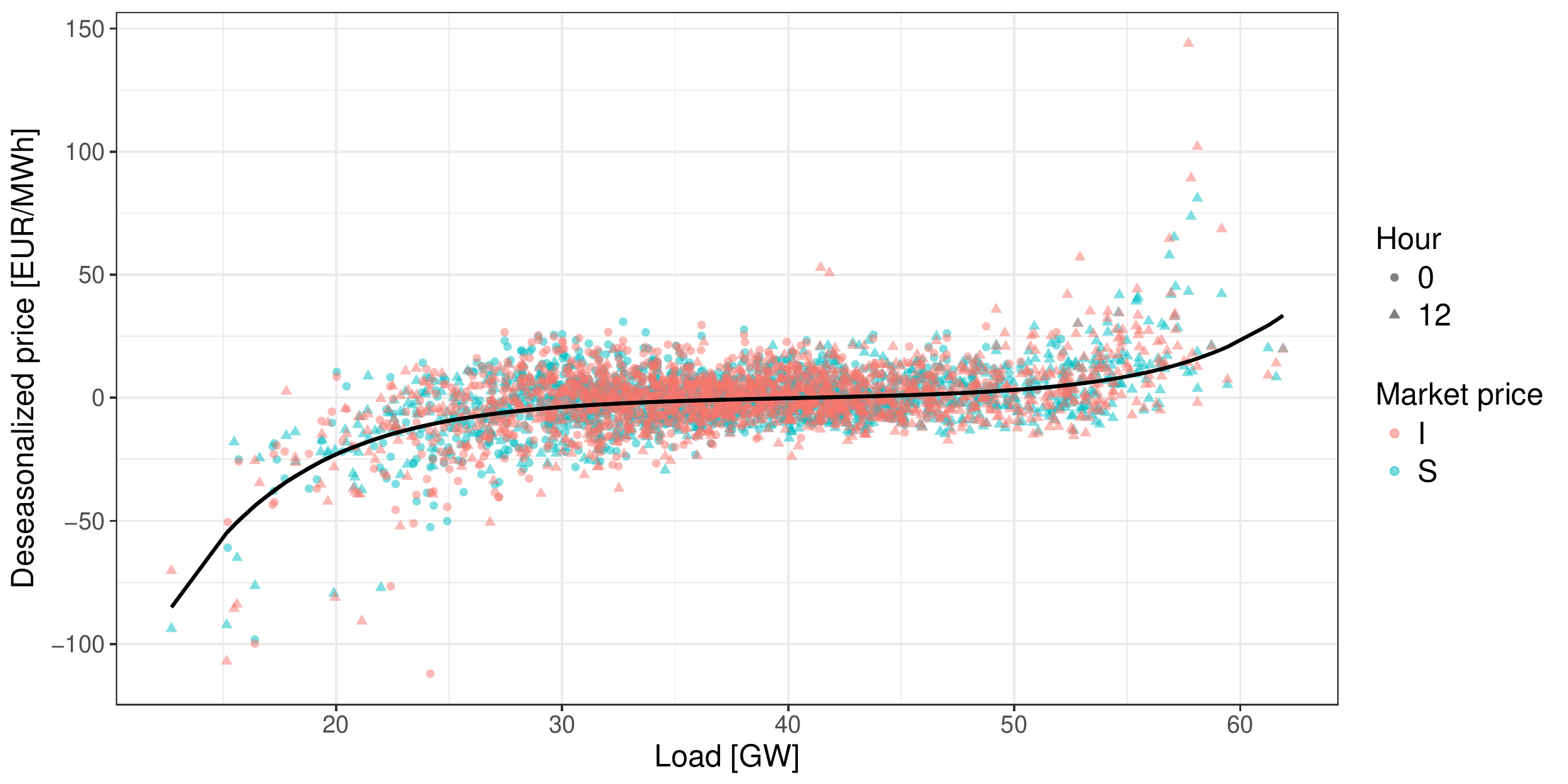}
\caption{Deseasonalized market intraday~$I^M- \gamma_3$ and day-ahead~$S^M- \gamma_3$ spot prices for the delivery hours 0--1 and 12--13 together with the intrinsic price curve $p- \gamma_3$~(black).}
\label{fig:EmpiricalIntrinsicPriceSupplyCurve}
\end{figure}

Analogously to the proof of Lemma~\ref{lemma:EmpiricalIntrinsicElectricityPriceTradablePrice} we can derive an explicit formula for the risk premium:
\[
\pi_t(\tau) = \left[\gamma_1(t; \tau) - \gamma_2(t; \tau)\right]  - \left[\tilde{\gamma}_1(t; \tau)  - \tilde{\gamma}_2(t; \tau) \right]
\]
for all $t \leq \tau_e$, if we define 
\[
\tilde{\gamma}_i(t; \tau) :=  e^{ \alpha_i \left(  g(\tau_e)  + e^{-\lambda \varepsilon} \left(1 - e^{-\lambda\tau}\right) \sigma \theta + e^{-\lambda(\tau_e - t)} \tilde{X}_t + \frac{\alpha_i \sigma^2}{4 \lambda} \left(1 - e^{-2\lambda(\tau_e - t)}\right) - \beta_i \right) },
\]
where $\tilde{X}$ is given in Remark~\ref{remark:EmpiricalSystemLoadUnderP}. Figure~\ref{fig:RiskPremiumPlotsIntrinsicPrice} illustrates the evolution of the risk premium through time. We see that we find an overall negative risk premium for all the plotted contracts, indicating that the \emph{``producers' desire to hedge their positions outweights that of the consumers''} \citep{BenthCartea2008}. In that sense our findings support the results of \citet{BenthCartea2008}, and not those of \citet{Redl2012,Viehmann2011}.

\begin{figure}[p]
\centering
\begin{subfigure}[c]{\textwidth}
\includegraphics[width=\textwidth]{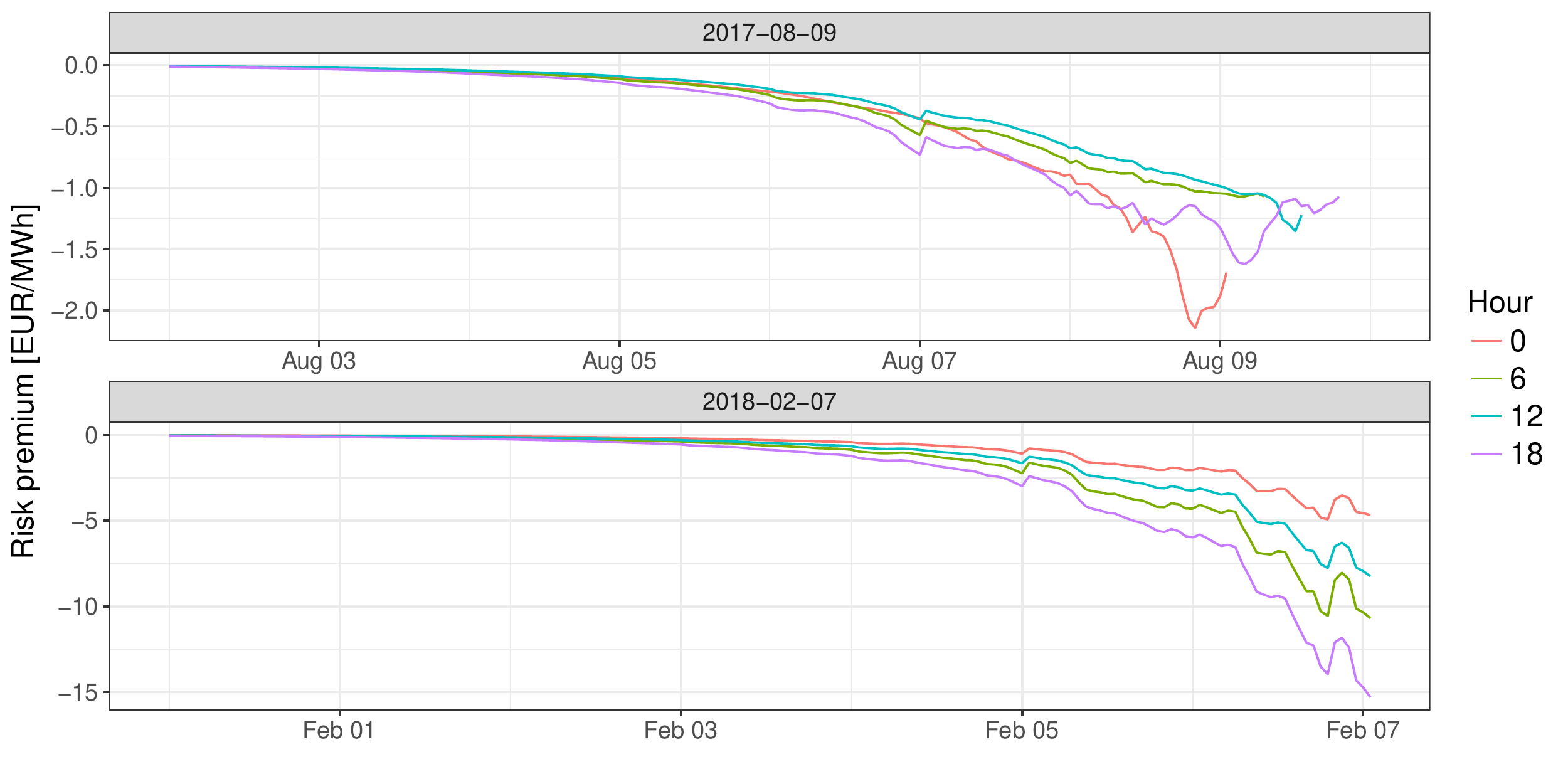}
\subcaption{Risk premia through time for two different delivery dates: the second Wednesday of August 2017 (Summer) and of February 2018 (Winter).}
\label{fig:RiskPremiumPlotsIntrinsicPrice}
\end{subfigure}

\begin{subfigure}[c]{\textwidth}
\includegraphics[width=\textwidth]{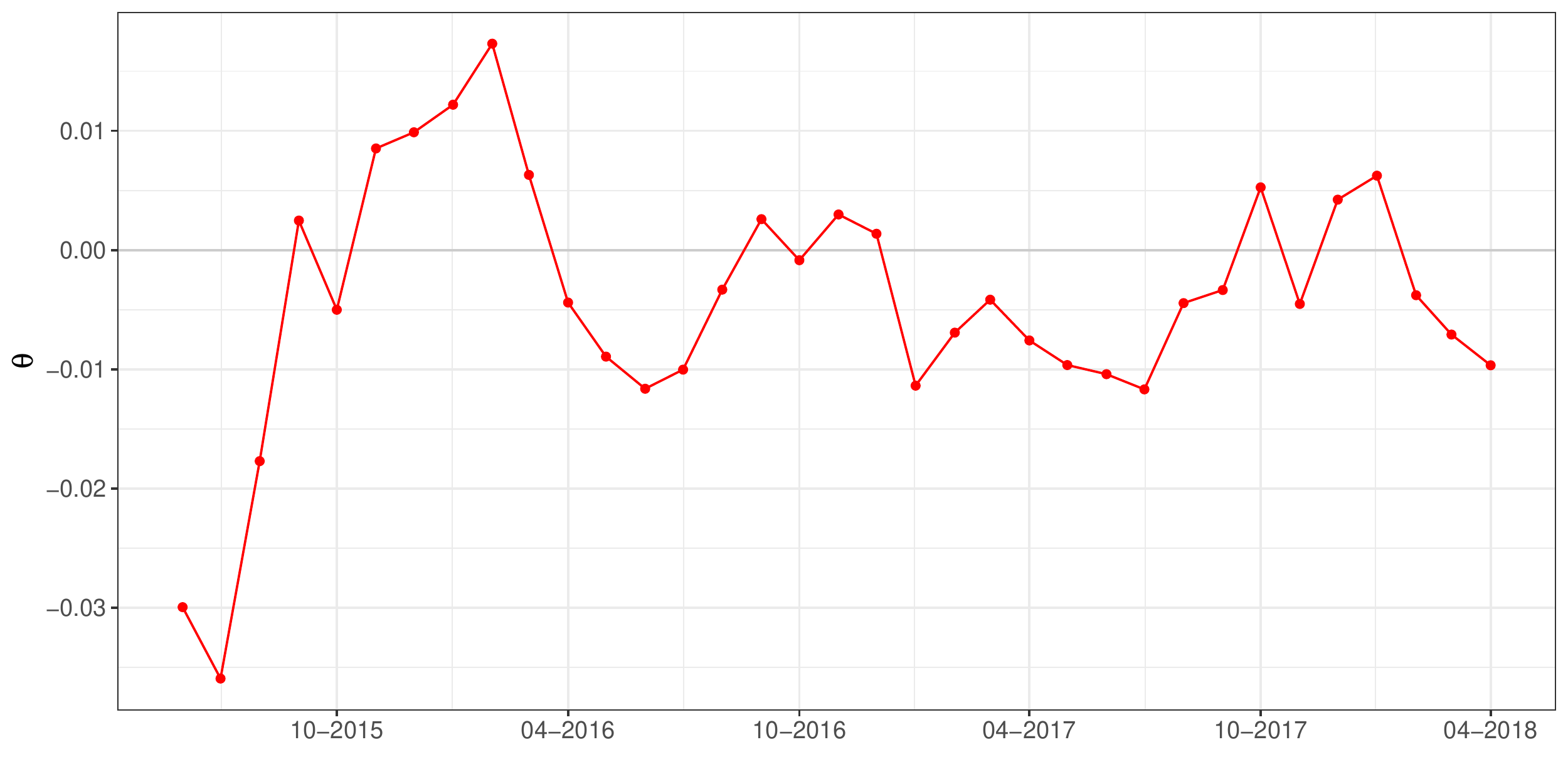}
\subcaption{Monthly implied Girsanov parameter~$\theta$.}
\label{fig:ImpliedTheta}
\end{subfigure}

\caption{Difference between the real-world measure~$P$ and the risk-neutral measure~$Q$.}
\end{figure}

In Figure~\ref{fig:ImpliedTheta} the implied Girsanov parameter~$\theta$ per month is shown. These were computed by solving Equation~\eqref{eq:OptimizationProblemForCalibration} for each month with the parameters~$\alpha_1$, $\alpha_2$, $\beta_1$, and $\beta_2$ fixed at the values we estimated before. The first thing we notice is the change in level from August 2015 to September 2015, where the value jumps from around -0.03 to approximately zero. We see that the implied~$\theta$ changes sign at least twice a year but is negative for most months (in 23 of the 35 months). We see that the positive values all occur during the months September to March. Furthermore, the implied~$\theta$ shows that the assumption of a constant value might have been an oversimplification, which should be investigated in future work.

\section{Conclusion}
In this paper we introduced a new concept for modelling electricity prices. We have discussed how this theory connects the classical theory of storage with the concept of a risk premium through the introduction of an unobservable intrinsic electricity price~$p(\tau)$. Since all tradable electricity contracts are derivatives of the actual intrinsic price, their prices should all be derived under the risk-neutral measure~$Q$. Based on this assumption we derived the prices for all common contracts such as the intraday spot price, the day-ahead spot price, and futures prices. Furthermore, we have shown how this framework relates to existing modelling approaches such as the Heath-Jarrow-Morton modelling approach, e.g. see \citet{Hinz2005,Kiesel2009,Latini2018,Hinderks2018,Benth2019}.

In the final part of this article we estimated a structural model from the difference between the intraday and day-ahead spot prices. By construction of this framework we could directly estimate the measure change between real-world measure~$P$ and the risk-neutral measure~$Q$. With this result we derived and computed the risk premium for several delivery times. We found that the risk premium is negative, indicating that the \emph{``producers' desire to hedge their positions outweights that of the consumers''} \citep{BenthCartea2008}.

For further research it is of interest to investigate the many possibilities for modelling the intrinsic electricity price and develop calibration methods that use all market data, i.e. from intraday, day-ahead spot, and futures markets, in the spirit of \citet{Caldana2017}. Existing models could be fitted to this framework and the results on the measure change could be investigated. In particular, the Girsanov parameter~$\theta$ could be made time-dependent. Finally, the framework as it is presented here is based on a probability space with the natural Brownian filtration. This setting could possibly be extended to a more general setting, in which also jump processes are allowed.

\section*{Acknowledgments}
WJH is grateful for the financial support from Fraunhofer ITWM (\emph{Fraunhofer Institute for Industrial Mathematics}, \url{www.itwm.fraunhofer.de}).

\appendix

\section{Options on futures}
Keeping in mind that the price generating process~$\varphi$ can also be used as one of the modelling ingredients, we can formulate results for the price of European call and put options in for two special cases of the price generating process, which yield normally or lognormally distributed prices.

For deterministic price generating processes we can find:
\begin{proposition}[Normal distribution] \label{proposition:DistributionFuturesPriceUnderQWithDeterministicPhi}
If  $\varphi(\tau)$ is deterministic process for all~$\tau$, then the conditional futures price~$F_t(\mathcal{T}) \, | \, \F_u$ is normally distributed under~$Q$ with mean
\[
\mu_{u} := F_u(\mathcal{T}) =   F_0(\mathcal{T})  +  \int_0^u\varphi_s(\mathcal{T})' \cdot dW_s
\]
and variance
\[
\sigma_{u,t}^2 := \int_u^t \varphi_s(\mathcal{T})' \cdot  \varphi_s(\mathcal{T}) \, ds
\]
for all $u \leq t \leq \tau_1 - \delta$.
\end{proposition}
\begin{proof}
For deterministic~$\varphi$ we know through its characteristic function that the integral~$\int_0^t \varphi_s(\mathcal{T})' \cdot dW_s$ is normally distributed with mean 0 and variance~$e^{-2rt} \sigma_{0,t}^2$. This is easily extended to any $u$.
\end{proof}

With the help of this proposition and the following auxilary variable
\[
\Delta_{u,t} := \frac{ F_u(\mathcal{T}) -  K}{ \sigma_{u,t}}
\]
we can compute the price of European put and call options on the futures price~$F_t(\mathcal{T})$.

\begin{lemma}[Call and put options]
If  $\varphi(\tau)$ is deterministic process for all~$\tau$, then for all $u\leq t \leq \tau_1 - \delta$ the price at time~$u$ of a European option with strike~$K$ on the futures contract~$F_t(\mathcal{T})$ is given by
\[
C_u(F_t(\mathcal{T}); K) =  e^{-r(t-u)}  \left(F_u(\mathcal{T}) - K \right) \Phi\left(\Delta_{u,t}\right) + \frac{ e^{-r(t-u)} \sigma_{u,t}}{\sqrt{2\pi}} e^{-\frac{1}{2}\Delta_{u,t}^2}
\]
for a call and by
\[
P_u(F_t(\mathcal{T}); K) =  e^{-r(t-u)}  \left( K - F_u(\mathcal{T}) \right) \Phi\left(-\Delta_{u,t}\right) + \frac{ e^{-r(t-u)} \sigma_{u,t}}{\sqrt{2\pi}} e^{-\frac{1}{2}\Delta_{u,t}^2}
\]
for a put option. Here $\Phi$ is the cumulative distribution function of the standard normal distribution.
\end{lemma}
\begin{proof}
Directly computing the conditional expectation yields
\begin{align*}
C_u(F_t(\mathcal{T}); K) &= \E_Q\left[ e^{-r(t-u)} \left(F_t(\mathcal{T})- K\right)^+ \, | \, \F_u\right] \\
&= e^{-r(t-u)}\, \E_Q\left[  \left(Y - K \right)^+ \,\Big| \, \F_u\right],
\end{align*}
where $Y$ is normally distributed with mean~$\mu_{u}$ and variance~$\sigma_{u,t}^2$ as given in Proposition~\ref{proposition:DistributionFuturesPriceUnderQWithDeterministicPhi}. Therefore, we compute
\[
C_u(F_t(\mathcal{T}); K) =  \frac{e^{-r(t-u)}}{\sqrt{2 \pi \sigma_{u,t}^2}}\, \int_K^\infty \left(y-K\right) e^{-\frac{1}{2} \frac{(y - \mu_{u})^2}{\sigma_{u,t}^2}} \, dy,
\]
from which the result follows by substitution of $y' = \frac{y- \mu_{u}}{\sigma_{u,t}}$. The proof follows analogously for put options.
\end{proof}

In contrast to Proposition~\ref{proposition:DistributionFuturesPriceUnderQWithDeterministicPhi} we can derive a lognormal distribution in the following case:
\begin{proposition}[Lognormal distribution]
If the price generating process is of the form
\begin{equation} \label{eq:IntrinsicElectricityPriceLognormalCondition}
\varphi_t(\tau) = \sigma_t \, f_t(\tau)
\end{equation}
for an $\R^d$-valued, deterministic, quadratic integrable process~$\sigma_t$ independent of the delivery time~$\tau$, then forward price is given by
\[
f_t(\tau) = f_0(\tau) \,  e^{- \frac{1}{2} \int_0^t \sigma_s' \cdot \sigma_s \, ds +  \int_0^t \sigma_s' \cdot dW_s},
\]
and, in particular, $f_t(\tau)$ has a lognormal distribution.
\end{proposition}
\begin{proof}
Follows directly from the SDE in Equation~\eqref{eq:IntrinsicElectricityPriceForwardSDE} and \citet[Chapter 5.6C]{Karatzas1998}.
\end{proof}

From the definition of the futures contract it follows immediately that:
\begin{corollary}
If the price generating process is of the form of Equation~\eqref{eq:IntrinsicElectricityPriceLognormalCondition}, then the futures price is given by
\[
F_t(\mathcal{T}) =  F_0(\mathcal{T}) \, e^{ - \frac{1}{2} \int_0^t \sigma_s' \cdot \sigma_s \, ds +  \int_0^t \sigma_s' \cdot dW_s}
\]
for all $ t \leq \tau_1 - \delta$ and has a lognormal distribution.
\end{corollary}

As for any lognormally distributed asset we can apply the Black-76 formula to derive the price of European call and put options \citep{Black1976}. Therefore, let us define the common auxiliary variables 
\[
d_\pm^{u,t} := \frac{\ln F_u(\mathcal{T}) - \ln K \pm \int_u^t \sigma_s' \cdot \sigma_s \, ds}{\sqrt{ \int_u^t \sigma_s' \cdot \sigma_s \, ds}}
\]
for any $u\leq t$. 

\begin{lemma}[Call and put options]
If the price generating process is of the form of Equation~\eqref{eq:IntrinsicElectricityPriceLognormalCondition}, then for all $u\leq t \leq \tau_1 - \delta$ the price at time~$u$ of a European option with strike~$K$ on the futures contract~$F_t(\mathcal{T})$ is given by
\[
C_u(F_t(\mathcal{T}); K) =    e^{-r(t-u)} \left[ F_u(\mathcal{T}) \, \Phi\left(d_+^{u,t}\right) - K \, \Phi\left(d_-^{u,t}\right)  \right]
\]
for call and by
\[
P_u(F_t(\mathcal{T}); K) = e^{-r(t-u)} \left[K \, \Phi\left(-d_-^{u,t}\right) -  F_u(\mathcal{T}) \, \Phi\left(-d_+^{u,t}\right)  \right]
\]
for put options. Here $\Phi$ is the cumulative distribution function of the standard normal distribution.
\end{lemma}

\bibliographystyle{abbrvnat}
\bibliography{references}
\end{document}